\documentclass[letterpaper,10pt,conference]{ieeeconf}

\usepackage{amsmath}
\usepackage{amssymb}
\usepackage{bm}

\usepackage[english]{babel}

\usepackage{graphicx}
\usepackage{psfrag}

\newcommand{\xlabel}[1]{\psfrag{xlabel}[Bc][Bc]{#1}} 
\newcommand{\ylabel}[1]{\psfrag{ylabel}[Bc][Bc]{#1}} 
\newcommand{\legend}[2]{\psfrag{#1}[Bl][Bl]{\raisebox{0.5pt}[0cm][0cm]{#2}}} 

\newtheorem{theorem}{Theorem}
\newtheorem{lemma}[theorem]{Lemma}

\newtheorem{definition}{Definition}
\newtheorem{assumption}{Assumption}

\newtheorem{remarkenv}{Remark}
\newenvironment{remark}{\begin{remarkenv}}{\hfill\raisebox{0.5mm}[0cm][0cm]{$\lhd$}\end{remarkenv}}

\newcommand{\sortbib}[1]{}

\usepackage{dsfont}
\newcommand{\R}{\ensuremath{\mathds{R}}} 

\newcommand{\ones}{\ensuremath{\bm{1}}} 

\newcommand{\defl}{\ensuremath{\mathrel{\mathop:}=}}              
\newcommand{\kron}{\ensuremath{\raisebox{0.7pt}{$\:\otimes\:$}}}  
\newcommand{\T}{\ensuremath{\mathrm{T}}}                          

\DeclareMathOperator{\diag}{diag}       
\DeclareMathOperator{\rank}{rank}        

\usepackage{xspace}

\newcommand{\Ecal}{\ensuremath{\mathcal{E}}}

\newcommand{\Gcal}{\ensuremath{\mathcal{G}}}

\newcommand{\Tcal}{\ensuremath{\mathcal{T}}}

\newcommand{\Vcal}{\ensuremath{\mathcal{V}}}

\newcommand{\sys}{\ensuremath{\bm{\Sigma}}}

\newcommand{\sysred}{\ensuremath{\bm{\hat{\Sigma}}}}

\IEEEoverridecommandlockouts
\usepackage{tikz}
\usetikzlibrary{arrows}

\renewcommand{\sc}{\ensuremath{\text{\rm c}}}
\newcommand{\se}{\ensuremath{\text{\rm e}}}

\renewcommand{\sf}{\ensuremath{\text{\rm f}}}
\newcommand{\so}{\ensuremath{\text{\rm o}}}
\newcommand{\su}{\ensuremath{\text{\rm u}}}

\title{\LARGE\bf
Model reduction of networked passive systems through clustering
}

\author{Bart Besselink, Henrik Sandberg, Karl Henrik Johansson%
\thanks{The authors are with the ACCESS Linaeus Centre and Department of Automatic Control, School of Electrical Engineering, KTH Royal Institute of Technology, Stockholm, Sweden, {\tt\small bart.besselink@ee.kth.se, hsan@kth.se, kallej@kth.se}}%
}

\begin{document}

\maketitle
\thispagestyle{empty}
\pagestyle{empty}

\begin{abstract}
In this paper, a model reduction procedure for a network of interconnected identical passive subsystems is presented. Here, rather than performing model reduction on the subsystems, adjacent subsystems are clustered, leading to a reduced-order networked system that allows for a convenient physical interpretation. The identification of the subsystems to be clustered is performed through controllability and observability analysis of an associated edge system and it is shown that the property of synchronization (i.e., the convergence of trajectories of the subsystems to each other) is preserved during reduction. The results are illustrated by means of an example.
\end{abstract}

\section{Introduction}
Electrical power grids, social networks and the internet and biological or chemical networks are examples of large-scale networks of interconnected dynamical (sub)systems, see, e.g., \cite{strogatz_2001}. Their large scale and complexity complicates the analysis or control of such networked systems, motivating the need for tools to obtain \emph{approximate} networked systems with lower complexity.

Model reduction techniques such as balanced truncation \cite{moore_1981} or optimal Hankel norm approximation \cite{glover_1984} provide methods for obtaining reduced-order approximations of large-scale systems \cite{book_antoulas_2005,besselink_2013b}, but are not directly suited for application to networked systems. Namely, the application of such methods typically does not preserve the interconnection structure, making the reduced-order models hard to interpret and potentially irrelevant for the design of distributed controllers. This paper therefore deals with the development of a dedicated reduction procedure for networked systems, based on the clustering of subsystems.

Despite the large interest in networked systems, the reduction of such systems has not received much attention in the literature. An exception is given by the work in~\cite{ishizaki_2014}, where a method for the clustering of subsystems is developed, considering subsystems that have scalar first-order dynamics. A different perspective is taken in \cite{monshizadeh_2013}, where networks of identical linear (higher-order) subsystems are considered. In~\cite{monshizadeh_2013}, reduction is performed on the basis of these subsystems only, thus leaving the interconnection structure untouched. It is noted that such reduction techniques for networked systems can be considered as a structure-preserving model reduction technique \cite{sandberg_2009}.

In the current paper, networks of identical linear subsystems are considered. However, rather than performing reduction of the individual subsystems, reduction is achieved by clustering neighbouring subsystems. This thus essentially represents a reduction of the interconnection topology, leading to a reduced-order interconnection structure that allows for a convenient physical interpretation. In particular, the subsystems are assumed to be passive and the interconnection topology is assumed to have a tree structure. For such systems, the importance of each edge in the interconnection structure (representing a coupling between subsystems) is studied through an analysis of its controllability and observability properties, hereby identifying pairs of adjacent vertices (subsystems) that are hard to steer apart or difficult to distinguish. Motivated by the method of balanced truncation and its extensions, these pairs of adjacent vertices will be clustered to obtain a reduced-order interconnection topology.

This analysis relies on two crucial aspects. First, the passivity property of the subsystems ensures that controllability and observability properties of the entire networked system can be decomposed into parts associated to the interconnection topology and the subsystem dynamics, respectively. Here, the former is used to identify important edges. Second, a novel factorization of the graph Laplacian describing the interconnection topology is exploited, which allows for the definition of an \emph{edge Laplacian} for weighted and directed graphs, hereby extending a result from \cite{zelazo_2011}. For tree structures, this factorization is shown to have desirable properties in the scope of model reduction through clustering.

Finally, it will be shown that the reduced-order networked system obtained by the clustering of subsystems preserves synchronization (i.e., the convergence of trajectories of the subsystems to each other) of the original networked system.

The remainder of this paper is organized as follows. After defining the problem in Section~\ref{sec_problemsetting}, the edge Laplacian and its relation to synchronization is discussed in Section~\ref{sec_edgelaplacian_sync}. The clustering-based model reduction procedure is introduced in Section~\ref{sec_modred} and illustrated by means of an example in Section~\ref{sec_example}. Finally, conclusions are stated in Section~\ref{sec_conclusions}.

\textit{Notation.} The field of real numbers is denoted by $\R$. Given a matrix $X\in\R^{n\times m}$, its entry in row $i$ and column $j$ is denoted as $(X)_{ij}$. The identity matrix of size $n$ is denoted as $I_n$, whereas $\ones_n$ denotes the vector of all ones of length $n$. The subscript $n$ is omitted when no confusion arises. Moreover, $e_i$ denotes the $i$-th column of $I_n$. Finally, $X\kron Y$ denotes the Kronecker product of the matrices $X$ and $Y$, whose definition and properties can be found in, e.g., \cite{brewer_1978}.

\section{Problem setting}\label{sec_problemsetting}
A network of identical subsystems $\sys_i$ is considered, of which a minimal realization can be written in the form
\begin{align}
\sys_i:\left\{\begin{array}{rclcl}
\dot{x}_i &=& Ax_i + Bv_i &=& (J-R)Qx_i + Bv_i, \\
z_i &=& Cx_i &=& B^{\T}Qx_i,
\end{array}\right.\label{eqn_sysi}
\end{align}
with $x_i\in\R^n$, $v_i,z_i\in\R^m$ and $i\in\{1,2,\ldots,\bar{n}\}$. The rightmost representation in (\ref{eqn_sysi}) is a so-called port-Hamiltonian form \cite{willems_1972b,book_vanderschaft_2000}, in which $Q = Q^{\T}\succ0$ characterizes the energy stored in $\sys_i$ as $V(x_i) = \tfrac{1}{2}x_i^{\T}Qx_i$. Next, $J = -J^{\T}$ is a skew-symmetric matrix and $R = R^{\T}\succcurlyeq0$ represents any internal dissipation. It is well-known that such a system is passive (see, e.g., \cite{willems_1972b} for a definition of passivity).

The subsystems $\sys_i$ as in (\ref{eqn_sysi}) are interconnected as
\begin{align}\textstyle
v_i = \sum_{j=1,j\neq i}^{\bar{n}} w_{ij}(z_j - z_i) + \sum_{j=1}^{\bar{m}}g_{ij}u_j,
\label{eqn_couplingi}
\end{align}
where $u_j\in\R^m$, $j\in\{1,2,\ldots,\bar{m}\}$ are the external inputs to the networked system. In (\ref{eqn_couplingi}), the weights $w_{ij}\in\R$ satisfying $w_{ij}\geq0$ represent the strength of the diffusive coupling between the subsystems, whereas $g_{ij}\in\R$ describe the distribution and strength of the external inputs amongst the subsystems. Similarly, external outputs are given by
\begin{align}\textstyle
y_i = \sum_{j=1}^{\bar{n}}h_{ij}z_j
\label{eqn_outputi}
\end{align}
with $y_i\in\R^m$, $i\in\{1,2,\dots,\bar{p}\}$. After defining $L$ as
\begin{align}
(L)_{ij} = \left\{\begin{array}{ll}
-w_{ij} &,\, i \neq j, \\
\sum_{j=1,j\neq i}^{\bar{n}} w_{ij}\! &,\, i = j,
\end{array}\right.\label{eqn_Ldef}
\end{align}
and collecting the parameters $g_{ij}$ and $h_{ij}$ as $G = \{g_{ij}\}$ and $H = \{h_{ij}\}$, respectively, the networked system given by (\ref{eqn_sysi}), (\ref{eqn_couplingi}) and (\ref{eqn_outputi}) can be written as
\begin{align}
\sys:\left\{\begin{array}{rcl}
\dot{x} &=& (I\kron A - L\kron BC)x + (G\kron B)u, \\
y &=& (H\kron C)x
\end{array}\right.\label{eqn_sys}
\end{align}
where $x^{\T} = [\begin{array}{cccc} x_1^{\T} & x_2^{\T} & \ldots & x_{\bar{n}}^{\T} \end{array}]$, $u^{\T} = [\begin{array}{cccc} u_1^{\T} & u_2^{\T} & \ldots & u_{\bar{m}}^{\T} \end{array}]$ and $y^{\T} = [\begin{array}{cccc} y_1^{\T} & y_2^{\T} & \ldots & y_{\bar{p}}^{\T} \end{array}]$.

The objective of this paper is to obtain a reduced-order version of the networked system (\ref{eqn_sys}) through the clustering of neighbouring subsystems, essentially creating a new interconnection structure of the form (\ref{eqn_couplingi}). A cluster is represented by a single subsystem, approximating the dynamics of a group of neighbouring subsystems in the original networked system. Consequently, the resulting reduced-order networked system is easy to interpret. Furthermore, this reduced-order networked system should preserve synchronization properties (i.e., the convergence of subsystem trajectories to each other) of the original system. Moreover, the input-output behavior of the reduced-order system should provide a good approximation of that of the original networked system.

\section{Edge Laplacian and synchronization}\label{sec_edgelaplacian_sync}
The interconnection (\ref{eqn_couplingi}) of the subsystems as characterized by $L$ as in (\ref{eqn_Ldef}) can be associated to a directed graph $\Gcal = (\Vcal,\Ecal)$ (see, e.g, \cite{book_godsil_2001,book_mesbahi_2010} for details on graph theory). Here, $\Vcal = \{1,2,\ldots,\bar{n}\}$ represents the set of vertices characterizing the subsystems and $\Ecal\subseteq\Vcal\times\Vcal$ gives the set of directed edges (or arcs) satisfying $(i,j)\in\Ecal$ if and only if $w_{ji}>0$.

Besides this directed graph $\Gcal$, an undirected version of the same graph is introduced as follows.
\begin{definition}
Let $\Gcal$ be a directed graph with vertex set $\Vcal$ and (directed) edge set $\Ecal$. Then, the undirected graph $\Gcal_{\su} = (\Vcal,\Ecal_{\su})$ with $(i,j)\in\Ecal_{\su}$ if and only if $w_{ij} + w_{ji} > 0$ is said to be the underlying undirected graph.
\end{definition}
The underlying undirected graph $\Gcal_{\su}$ thus has an edge between vertices $i$ and $j$ if at least one of the weights $w_{ij}$ and $w_{ji}$ is strictly positive, i.e., if there exists at least one directed edge between $i$ and $j$. Then, by exploiting the incidence matrix $E\in\R^{\bar{n}\times\bar{n}_{\se}}$ (with elements in $\{0,\pm1\}$ and where $\bar{n}_{\se}$ is the number of edges in $\Gcal_{\su}$) of an arbitrary orientation of $\Gcal_{\su}$, the matrix $L$ as in (\ref{eqn_Ldef}) can be factorized as follows.
\begin{lemma}\label{lem_Lfactor}
Consider the matrix $L$ as in (\ref{eqn_Ldef}) and let $E$ be an oriented incidence matrix of the underlying undirected graph $\Gcal_{\su}$. Then, $L$ can be factored as
\begin{align}
L = FE^{\T},\label{eqn_Lfactor}
\end{align}
where $F$ has the same structure as $E$. In particular, let the $l$-th column of $E$ be given as $e_i-e_j$, i.e., characterizing the edge connecting vertices $i$ and $j$. Then, the $l$-th column of $F$ is given as $w_{ij}e_i - w_{ji}e_j$, with $w_{ij}$ the weights as in (\ref{eqn_couplingi}).
\end{lemma}
\begin{proof}
It is noted that the matrix $L$ as in (\ref{eqn_Ldef}) can be written as the sum of matrices characterizing each edge in $\Gcal_{\su}$ individually, leading to $L = \sum_{(i,j)\in\Ecal_{\su}} L_{ij}$. Here,
\begin{align}
L_{ij} = (w_{ij}e_i - w_{ji}e_j)(e_i-e_j)^{\T},
\end{align}
such that choosing the columns of $F$ and $E$ as $(w_{ij}e_i-w_{ji}e_j)$ and $e_i-e_j$, respectively, leads to (\ref{eqn_Lfactor}).
\end{proof}
The eigenvalues of $L$ can be related to graph-theoretical properties by exploiting the notion of a directed rooted spanning tree, which is defined as follows (see \cite{book_godsil_2001,ren_2005b}).
\begin{definition}\label{def_dirrootedspanningtree}
A graph $\Tcal$ is said to be a directed rooted spanning tree if it is a directed tree connecting all vertices of the graph, where every vertex, except the single root vertex, has exactly one incoming directed edge.
\end{definition}
The following result can be found in \cite{ren_2005b,book_mesbahi_2010}.
\begin{lemma}\label{lem_Lvsdirrootedspanningtree}
Consider the matrix $L$ as in (\ref{eqn_Ldef}) with $w_{ij}\geq0$. Then, $L$ has at least one zero eigenvalue and all nonzero eigenvalues are in the open right-half plane. Moreover, $L$ has exactly one zero eigenvalue if and only if the associated graph $\Gcal$ contains a directed rooted spanning tree as a subgraph.
\end{lemma}
Now, the property of containing a directed rooted spanning tree as a subgraph can be related to the matrix $F$ in the factorization (\ref{eqn_Lfactor}) when the underlying undirected graph $\Gcal_{\su}$ is a tree, as stated in the following lemma.
\begin{lemma}\label{lem_Fvsdirrootedspanningtree}
Let the graph $\Gcal$ characterized by $L$ as in (\ref{eqn_Ldef}) be such that the underlying undirected graph $\Gcal_{\su}$ is a tree. Then, $\rank F = \bar{n}-1$ if and only if\/ $\Gcal$ contains a directed rooted spanning tree as a subgraph.
\end{lemma}
\begin{proof}
First, it is noted that, as $\Gcal_{\su}$ is a tree, it has $\bar{n}-1$ edges and $F\in\R^{\bar{n}\times(\bar{n}-1)}$ (see also the definition of $F$ in the statement of Lemma~\ref{lem_Lfactor}). Consequently, the rank of $F$ is at most $\bar{n}-1$. Assume, for the sake of establishing a contradiction, that $\rank F = \bar{n}-c$ with $c>1$. Then, by the rank-nullity theorem, the null space of $F^{\T}$ has dimension $c$. If $V\in\R^{\bar{n}\times c}$ is a basis for this null space, it satisfies $F^{\T}V=0$. Evaluation of the product $V^{\T}L = V^{\T}FE^{\T} = 0$ implies that $V$ is also in the null space of $L^{\T}$, such that $\rank L < \bar{n}-c < \bar{n}-1$. However, that contradicts the assumption of the existence of a directed rooted spanning tree as a subgraph via Lemma~\ref{lem_Lvsdirrootedspanningtree}, such that $\rank F = \bar{n}-1$.

To prove the converse, assume that $\rank F = \bar{n}-1$. Also, as $\Gcal_{\su}$ is a tree, $\rank E = n-1$. Then, $L = FE^{\T}$ represents a full-rank factorization (see \cite{book_horn_1990}) and $\rank L = \bar{n}-1$. Consequently, by Lemma~\ref{lem_Lvsdirrootedspanningtree}, $\Gcal$ contains a directed rooted spanning tree as a subgraph.
\end{proof}

In the remainder of the paper, networks with a tree structure will be considered.
\begin{assumption}\label{ass_coupling}
The interconnection structure characterized by $L$ as in (\ref{eqn_Ldef}) is such that:
\begin{enumerate}
  \item the underlying undirected graph $\Gcal_{\su}$ is a tree, i.e., $\bar{n}_{\se} = \bar{n}-1$;
  \item the graph $\Gcal$ contains a directed rooted spanning tree as a subgraph.
\end{enumerate}
\end{assumption}
Here, it is remarked neither of these items implies the other.

Under Assumption~\ref{ass_coupling}, the following lemma holds.
\begin{lemma}\label{lem_Le_eigenvalues}
Let the interconnection structure characterized by $L$ as in (\ref{eqn_Ldef}) satisfy Assumption~\ref{ass_coupling} and consider its factorization (\ref{eqn_Lfactor}). Then, the matrix
\begin{align}
L_{\se} = E^{\T}F,
\label{eqn_Ledef}
\end{align}
which will be referred to as the \emph{edge Laplacian}, has all eigenvalues in the open right-half plane. Moreover, these eigenvalues equal the nonzero eigenvalues of $L$.
\end{lemma}
\begin{proof}
To prove this lemma, introduce the matrix $T$ as
\begin{align}
T = \left[\begin{array}{c} \nu^{\T} \\ E^{\T} \end{array}\right],
\label{eqn_T}
\end{align}
where $\nu$ is the left eigenvector for the zero eigenvalue of $L$, i.e., $\nu^{\T}L=0$. By the second item of Assumption~\ref{ass_coupling} and Lemma~\ref{lem_Lvsdirrootedspanningtree}, this eigenvalue has multiplicity one, such that $\nu$ is unique (up to scaling). Also, by exploiting Perron-Frobenius theory (see, e.g., \cite{book_horn_1990}), it can be shown that all elements of $\nu$ are nonnegative (and $\nu\neq0$). It is therefore assumed that $\nu$ is scaled such that $\nu^{\T}\ones = 1$. As each column of $E$ has only zero elements except for the pair $(1,-1)$, $\nu$ is linearly independent of the columns of $E$ and $T$ as in (\ref{eqn_T}) is nonsingular. Thus, its inverse exists. In particular, it is given as
$T^{-1} = [\begin{array}{cc} \ones & F(E^{\T}F)^{-1} \end{array}]$. Then, the application of the similarity transformation $T$ to $L$ as in (\ref{eqn_Ldef}) leads to
\begin{align}
TLT^{-1} = \left[\begin{array}{cc} 0 & 0 \\ 0 & E^{\T}F \end{array}\right] =
\left[\begin{array}{cc} 0 & 0 \\ 0 & L_{\se} \end{array}\right].
\label{eqn_proof_lem_Le_eigenvalues}
\end{align}
By Assumption~\ref{ass_coupling} and Lemma~\ref{lem_Lvsdirrootedspanningtree}, $L$ contains only a single zero eigenvalue, which is isolated from the matrix $L_{\se}$ in the representation (\ref{eqn_proof_lem_Le_eigenvalues}). Consequently, $L_{\se}$ contains all non-zero eigenvalues of $L$, which are in the open right-half plane by Lemma~\ref{lem_Lvsdirrootedspanningtree}.
\end{proof}
\begin{remark}\label{rem_edgedynamics}
The matrix $L_{\se}$ in (\ref{eqn_Ledef}) is directly related to the dynamics on the edges of the networked system (\ref{eqn_sys}). To show this, the edge coordinates $x_{\se} = (E^{\T}\kron I_n)x$ are introduced, representing the difference between the state components of two neighboring subsystems. By exploiting the networked dynamics (\ref{eqn_sys}), it is readily shown that $x_{\se}$ satisfies
\begin{align}
\dot{x}_{\se} &= (I_{\bar{n}-1}\kron A - L_{\se}\kron BC)x_{\se} + (E^{\T}G\kron B)u, \label{eqn_edgedynamics_nooutput}
\end{align}
where the factorization $L=FE^{\T}$ and the definition of $L_{\se}$ as in (\ref{eqn_Ledef}) is used. 
Motivated by its role in the dynamics (\ref{eqn_sys}), the matrix $L_{\se}$ might be thought of as the (directed and weighted) edge Laplacian for the graph $\Gcal$. The edge Laplacian for unweighted and undirected graphs is studied in~\cite{zelazo_2011}.
\end{remark}

The edge Laplacian $L_{\se}$, as introduced in Lemma~\ref{lem_Le_eigenvalues}, can be exploited to study synchronization of the networked system $\sys$ as in (\ref{eqn_sys}), as stated in the following theorem.
\begin{theorem}\label{thm_sys_sync}
Consider the networked system $\sys$ as in (\ref{eqn_sys}) with passive subsystems $\sys_i$ as in (\ref{eqn_sysi}). Moreover, let the interconnection structure characterized by $L$ as in (\ref{eqn_Ldef}) be such that Assumption~\ref{ass_coupling} holds. Then, any trajectory of\/ $\sys$ for $u=0$ satisfies (for all $i,j\in\Vcal$)
\begin{align}
\lim_{t\rightarrow\infty} \big(x_i(t) - x_j(t)\big) = 0.
\label{eqn_synchronizationdef}
\end{align}
\end{theorem}
\vspace{2mm}
\begin{proof}
In order to prove the theorem, it is remarked that synchronization as in (\ref{eqn_synchronizationdef}) can be equivalently characterized by asymptotic stability of the edge dynamics as in (\ref{eqn_edgedynamics_nooutput}) for $u=0$. Therefore, asymptotic stability of the edge dynamics (\ref{eqn_edgedynamics_nooutput}) will be shown by introduction of the candidate Lyapunov function
\begin{align}
V(x_{\se}) = x_{\se}^{\T}(K\kron Q)x_{\se},
\label{eqn_thm_sync_proof_lyapfunc}
\end{align}
where $Q = Q^{\T}\succ0$ is the energy function of the passive subsystems as in (\ref{eqn_sysi}). The matrix $K = K^{\T}\succ0$ in (\ref{eqn_thm_sync_proof_lyapfunc}) will be specified later. It is noted that, as both $K$ and $Q$ are symmetric positive definite, their Kronecker product $K\kron Q$ is symmetric positive definite as well. The time-differentiation of (\ref{eqn_thm_sync_proof_lyapfunc}) along the trajectories of (\ref{eqn_edgedynamics_nooutput}) for $u=0$, hereby using properties of the Kronecker product as well as $J = -J^{\T}$ leads to
\begin{align}
\!\!\dot{V}(x_{\se}) &= x_{\se}^{\T}\left( -2K\kron QRQ - (L_{\se}^{\T}K + KL_{\se})\kron C^{\T}C \right)x_{\se}, \nonumber \\
&\leq x_{\se}^{\T}(I\!\kron\! C^{\T})\!\left( (-L_{\se}^{\T}K - KL_{\se})\kron I \right)\!(I\!\kron\!C)x_{\se},\!\! \label{eqn_thm_sync_proof_Vdot_step2}
\end{align}
where the property $R = R^{\T}\succcurlyeq0$ is used to obtain the latter inequality. As Assumption~\ref{ass_coupling} guarantees that $-L_{\se}$ is Hurwitz (through Lemma~\ref{lem_Le_eigenvalues}), there exists a matrix $K=K^{\T}\succ0$ such that $L_{\se}^{\T}K + KL_{\se}\succ\alpha I$ for some $\alpha>0$. Then, (\ref{eqn_thm_sync_proof_Vdot_step2}) satisfies
\begin{align}
\dot{V}(x_{\se}) \leq -\alpha x_{\se}^{\T}(I\kron C^{\T})(I\kron C)x_{\se},
\end{align}
and asymptotic stability of the edge dynamics (\ref{eqn_edgedynamics_nooutput}) follows from observability of the subsystems $\sys_i$ (through minimality) and LaSalle's invariance principle (see, e.g., \cite{book_vanderschaft_2000}).
\end{proof}

\section{Model reduction through clustering}\label{sec_modred}

\subsection{Edge controllability and observability}
The reduction of the networked system $\sys$ as in (\ref{eqn_sys}) will be performed by clustering adjacent vertices (subsystems). To identify the vertices to be clustered, the importance of the edges connecting vertices will be analyzed. Thereto, the \emph{edge system} is introduced as
\begin{align}
\sys_{\se}:\left\{\begin{array}{rcl}
\dot{x}_{\se} &=& \left(I \kron A - L_{\se}\kron BC\right)x_{\se} + (G_{\se}\kron B)u \\
y_{\se} &=& (H_{\se}\kron C)x_{\se},
\end{array}\right.\label{eqn_syse}
\end{align}
with $x_{\se} = (E^{\T}\kron I)x\in\R^{\bar{n}-1}$, $G_{\se} = E^{\T}G$ and $H_{\se} = HF(E^{\T}F)^{-1}$. Also, it is convenient to introduce a different realization of (\ref{eqn_syse}), leading to the \emph{dual edge system} as
\begin{align}
\sys_{\sf}:\left\{\begin{array}{rcl}
\dot{x}_{\sf} &=& \left(I \kron A - L_{\se}\kron BC\right)x_{\sf} + (G_{\sf}\kron B)u \\
y_{\se} &=& (H_{\sf}\kron C)x_{\sf},
\end{array}\right.\label{eqn_sysf}
\end{align}
with $x_{\sf}=((E^{\T}F)^{-1}\kron I)x_{\se}$, $G_{\sf} = (E^{\T}F)^{-1}G_{\se}$ and $H_{\sf} = HF$.

Motivated by the well-known reduction method of balanced truncation \cite{moore_1981}, the importance of edges will be characterized through their controllability and observability properties, motivating the following definition.
\begin{definition}\label{def_edgegrams}
The matrices $\bar{P}_{\se}$ and $\bar{Q}_{\sf}$ are said to be the edge controllability Gramian and edge observability Gramian of the system $\sys$ as in (\ref{eqn_sys}) if they are the controllability Gramian of $\sys_{\se}$ as in (\ref{eqn_syse}) and the observability Gramian of $\sys_{\sf}$ as in (\ref{eqn_sysf}), respectively.
\end{definition}
As the edge Gramians in Definition~\ref{def_edgegrams} do not necessarily allows for an insightful characterization of the importance of \emph{individual} edges, the following matrices are introduced.
\begin{definition}\label{def_genedgegrams}
The matrices $\tilde{P}_{\se} = \tilde{\Pi}^{\sc}\kron Q^{-1}$ and $\smash{\tilde{Q}_{\sf}} = \smash{\tilde{\Pi}^{\so}\kron Q}$ are, respectively, said to be a generalized edge controllability Gramian and a generalized edge observability Gramian for the networked system $\sys$ as in (\ref{eqn_sys}) if the matrices $\tilde{\Pi}^{\sc}\succcurlyeq0$ and $\tilde{\Pi}^{\so}\succcurlyeq0$ are diagonal and satisfy the inequalities
\begin{align}
L_{\se}\tilde{\Pi}^{\sc} + \tilde{\Pi}^{\sc}L_{\se}^{\T} - E^{\T}GG^{\T}E &\succcurlyeq 0,
\label{eqn_genedgecongram_ineq}\\
L_{\se}^{\T}\tilde{\Pi}^{\so} + \tilde{\Pi}^{\so}L_{\se} - F^{\T}H^{\T}HF &\succcurlyeq 0.
\label{eqn_genedgeobsgram_ineq}
\end{align}
\end{definition}
\vspace{2mm}
The introduction of the generalized edge Gramians allows for the interpretation of controllability and observablity properties on the basis of the interconnection topology only, thus providing a suitable basis for a reduction procedure based on clustering of vertices (i.e., subsystems). In particular, the (structured) generalized edge Gramians provide an upper bound on the real Gramians, as formalized next.
\begin{theorem}\label{thm_genedgegrams_bounds}
Consider the networked system $\sys$ as in (\ref{eqn_sys}) satisfying Assumption~\ref{ass_coupling} and assume that the generalized edge controllability Gramian and generalized edge observability Gramian as in Definition~\ref{def_genedgegrams} exist. Then, they bound the edge controllability Gramian and edge observability Gramian as in Definition~\ref{def_edgegrams} as
\begin{align}
\bar{P}_{\se} &\preccurlyeq \tilde{\Pi}^{\sc}\kron Q^{-1}, \label{eqn_thm_genedgegrams_Pebound} \\
\bar{Q}_{\sf} &\preccurlyeq \tilde{\Pi}^{\so}\kron Q. \label{eqn_thm_genedgegrams_Qebound}
\end{align}
\end{theorem}
\vspace{2mm}
\begin{proof}
The proof is inspired by results from \cite{besselink_2013d} (see \cite{koshita_2006} for a similar result). In particular, the controllability case will be proven, as the observability case follows similarly. First, it is remarked that $\sys_{\se}$ is asymptotically stable, as follows from Assumption~\ref{ass_coupling} and Theorem~\ref{thm_sys_sync}. As a result, the edge controllability Gramian $\bar{P}_{\se}$ can be obtained as the unique solution of a Lyapunov equation, see, e.g., \cite{book_antoulas_2005}. Now, consider the matrix
\begin{align}
\Lambda &\defl \left(I \kron A - L_{\se}\kron BC\right)(\tilde{\Pi}^{\sc}\kron Q^{-1}) \nonumber\\
&\phantom{\defl}\quad + (\tilde{\Pi}^{\sc}\kron Q^{-1})\left(I \kron A - L_{\se}\kron BC\right){}^{\!\T} \nonumber\\
&\phantom{\defl}\quad + (G_{\se}\kron B)(G_{\se}\kron B)^{\T},
\label{eqn_thm_edgecongram_decomp_proof_step2}
\end{align}
which is of the same form as this Lyapunov equation. Consequently, if $\Lambda\preccurlyeq0$, the matrix $\tilde{\Pi}^{\sc}\kron Q^{-1}$ satisfies the corresponding Lyapunov inequality and the desired result follows \cite{book_dullerud_2000}. The substitution of the relations for the system matrices as in (\ref{eqn_sysi}) in (\ref{eqn_thm_edgecongram_decomp_proof_step2}), hereby exploiting properties of the Kronecker product \cite{brewer_1978}, leads to
\begin{align}
\!\!\Lambda = -2(\tilde{\Pi}^{\sc}\kron R) - \big(L_{\se}\tilde{\Pi}^{\sc} + \tilde{\Pi}^{\sc} L_{\se}^{\T} - G_{\se}G_{\se}^{\T}\big)\kron BB^{\T}\!.\!\!
\label{eqn_thm_edgecongram_decomp_proof_step4}
\end{align}
Then, by recalling that $G_{\se} = E^{\T}G$ and that (\ref{eqn_genedgecongram_ineq}) holds, it follows that $\Lambda\preccurlyeq0$, proving the theorem.
\end{proof}
\begin{remark}
From (\ref{eqn_thm_genedgegrams_Pebound})-(\ref{eqn_thm_genedgegrams_Qebound}), it is clear that the tightest bound is obtained when the solutions $\tilde{\Pi}^{\sc}$ of (\ref{eqn_genedgecongram_ineq}) an $\tilde{\Pi}^{\so}$ of (\ref{eqn_genedgeobsgram_ineq}) are minimized in some sense. A suitable heuristic is the minimization of the trace of $\tilde{\Pi}^{\sc}$ and $\tilde{\Pi}^{\so}$.
\end{remark}

\subsection{One-step reduction by clustering}
The matrices $\tilde{\Pi}^{\sc}$ and $\tilde{\Pi}^{\so}$ as in Definition~\ref{def_genedgegrams} are written as
\begin{align}
\tilde{\Pi}^{\sc} &= \diag\{\pi_1^{\sc},\pi_2^{\sc},\ldots,\pi_{\bar{n}_{\se}}^{\sc}\},
\label{eqn_genedgecongram_diagonal} \\
\tilde{\Pi}^{\so} &= \diag\{\pi_1^{\so},\pi_2^{\so},\ldots,\pi_{\bar{n}_{\se}}^{\so}\},
\label{eqn_genedgeobsgram_diagonal}
\end{align}
where the ordering $\pi_i^{\sc}\pi_i^{\so}\geq\pi_{i+1}^{\sc}\pi_{i+1}^{\so}$ is assumed. It is noted that this can always be achieved by a suitable permutation of the edge numbers. As $\tilde{\Pi}^{\sc}$ and $\tilde{\Pi}^{\so}$ characterize the controllability and observability properties of edges, respectively, the products $\pi_i^{\sc}\pi_i^{\so}$ provide a characterization of the importance of each edge. Consequently, the final edge in the edge system is assumed to be the least important.

Assuming that the vertices $i$ and $j$ associated to the least important edge are numbered as $i = \bar{n}-1$ and $j = \bar{n}$ (this can again be achieved by a suitable permutation of vertex numbers), the projection matrices
\begin{align}
V = \left[\begin{array}{cc} I & 0 \\ 0 & 1 \\ 0 & 1 \end{array}\right], \quad
W = \left[\begin{array}{cc} I & 0 \\ 0 & \frac{w_{ji}}{w_{ij}+w_{ji}} \\ 0 & \frac{w_{ij}}{w_{ij}+w_{ji}} \end{array}\right],
\label{eqn_projection_VW}
\end{align}
are introduced. Then, the approximation of the state $x$ of $\sys$ as in (\ref{eqn_sys}) as $x\approx(V\kron I)\xi$ and projection of the resulting dynamics by $W\kron I$ leads to the one-step clustered system~as
\begin{align}
\sysred:\left\{\begin{array}{rcl}
\dot{\xi} &=& (I\kron A - \hat{L}\kron BC)\xi + (\hat{G}\kron B)u, \\
\hat{y} &=& (\hat{H}\kron C)\xi,
\end{array}\right.\label{eqn_sysred}
\end{align}
with $\hat{L} = W^{\T}LV$, $\hat{G} = W^{\T}G$ and $\hat{H} = HV$.

In order to analyze the properties of the reduced-order networked system (\ref{eqn_sysred}), the matrices $E$ and $F$ as in (\ref{eqn_Lfactor}) can be partitioned according to the clustered vertices $i = \bar{n}-1$ and $j=\bar{n}$ and the corresponding edge $l = \bar{n}_{\se}$ as
\begin{align}
E = \left[\begin{array}{cc}
E_{00} & 0 \\ E_{i0} & E_{il} \\ E_{j0} & E_{jl}
\end{array}\right], \quad
F = \left[\begin{array}{cc}
F_{00} & 0 \\ F_{i0} & F_{il} \\ F_{j0} & F_{jl}
\end{array}\right].
\label{eqn_partitioning_EF}
\end{align}
Here, it is noted that the zero entries in both $E$ and $F$ result from the fact that the corresponding column represents the edge connecting vertex $i$ to $j$. Specifically, $E_{il}\in\{-1,1\}$ and $E_{jl} = -E_{il}$. Similarly, $F_{il} = w_{ij}E_{il}$ and $F_{jl} = w_{ji}E_{jl}$, as follows from Lemma~\ref{lem_Lfactor}.

\begin{lemma}\label{lem_Lhat_factor}
Let the interconnection structure characterized by $L$ as in (\ref{eqn_Ldef}) satisfy Assumption~\ref{ass_coupling} and consider its factorization (\ref{eqn_Lfactor}), in which $E$ and $F$ are partitioned as in (\ref{eqn_partitioning_EF}). Moreover, let $\hat{L}$ be the reduced-order interconnection matrix obtained by projection using the matrices (\ref{eqn_projection_VW}), let $\hat{\Gcal}$ be the graph on $\bar{n}-1$ vertices it characterizes and $\hat{\Gcal}_{\su}$ the underlying undirected graph. Then,
\begin{enumerate}
  \item the matrix $\hat{L}$ can be factored as $\hat{L} = \hat{F}\hat{E}^{\T}$ with
\begin{align}
\hspace{-7mm}
\hat{E} = \left[\!\begin{array}{c} E_{00} \\ E_{i0} \!+\! E_{j0} \end{array}\!\right]\!,
\hat{F} = \left[\begin{array}{c} F_{00} \\ \frac{w_{ji}}{w_{ij}+w_{ji}}F_{i0} \!+\! \frac{w_{ij}}{w_{ij}+w_{ji}}F_{j0} \end{array}\!\right]\!\!;\!\!
\label{eqn_lem_Lhat_factor_Ehat_Fhat}
\end{align}
\item the underlying undirected graph $\hat{\Gcal}_{\su}$ is a tree;
\item the graph $\hat{\Gcal}$ contains a directed rooted spanning tree as a subgraph.
\end{enumerate}
\end{lemma}
\begin{proof}
The first item can be proven by exploiting the factorization of $L$ as in (\ref{eqn_Lfactor}), from which it follows that
\begin{align}
W^{\T}LV = W^{\T}FE^{\T}V = (W^{\T}F)(V^{\T}E)^{\T}.
\end{align}
The computation of $V^{\T}E$, hereby using (\ref{eqn_partitioning_EF}), leads to
\begin{align}
V^{\T}E = \left[\begin{array}{cc} E_{00} & 0 \\ E_{i0} + E_{j0} & 0 \end{array}\right],
\label{eqn_lem_Lhat_factor_proof_projE}
\end{align}
where it is noted that the final column contains all zeros since $E_{il} + E_{jl} = 0$. Namely, the $l$-th (with $l=\bar{n}-1$) column of $E$ characterizes the edge that connects vertices $i$ and $j$, such that $E_{il}\in\{1,-1\}$ and $E_{jl} = -E_{il}$. Similarly, it can be shown that
\begin{align}
W^{\T}F = \left[\begin{array}{cc} F_{00} & 0 \\ \frac{w_{ji}}{w_{ij}+w_{ji}}F_{i0} + \frac{w_{ij}}{w_{ij}+w_{ji}}F_{j0} & 0 \end{array}\right],
\label{eqn_lem_Lhat_factor_proof_projF}
\end{align}
hereby using a similar argument to proof that the final column contains all zeros. Herein, the choice of the weights in $W$ as in (\ref{eqn_projection_VW}) is crucial. Finally, setting $\hat{E}$ and $\hat{F}$ as the nonzero columns of $V^{\T}E$ and $W^{\T}F$, respectively, proves the first item in the statement of the theorem.

To prove the second item, it is noted that the only nonzero elements in each column in $\hat{E}$ have values given by the pair $(1,-1)$, which follows from the properties of the original indicence matrix $E$ as in (\ref{eqn_partitioning_EF}) and the definition of $\hat{E}$ in (\ref{eqn_lem_Lhat_factor_Ehat_Fhat}). Next, it is clear that $\hat{E}$ has $\bar{k} = \bar{n}-1$ rows, corresponding to the vertices of the clustered graph, and $\bar{k}-1$ columns, corresponding to the edges in the underlying undirected graph $\Gcal_{\su}$. Moreover, it can be concluded from (\ref{eqn_lem_Lhat_factor_proof_projE}) that $\rank\hat{E} = \bar{k}-1$, such that $\Gcal_{\su}$ is connected. As a tree is the only graph that connects $\bar{k}$ vertices with $\bar{k}-1$ edges, $\Gcal_{\su}$ is a tree.

The third item can be proven using similar arguments. It can be observed that $\hat{F}$ has the same size and structure (in the sense as in the statement of Lemma~\ref{lem_Lfactor}) as $\hat{E}$. Next, it follows from (\ref{eqn_lem_Lhat_factor_proof_projF}) that $\rank\hat{F} = \bar{k}-1$, such that the result follows from Lemma~\ref{lem_Fvsdirrootedspanningtree}.
\end{proof}

To obtain further properties of the one-step clustered model $\sysred$ as in (\ref{eqn_sysred}), the corresponding edge system is considered, hereby exploiting the explicit expression of $\hat{E}$ in (\ref{eqn_lem_Lhat_factor_Ehat_Fhat}) to define $\xi_{\se} = (\hat{E}^{\T}\kron I)\xi$, leading to
\begin{align}
\sysred_{\se}:\left\{\begin{array}{rcl}
\dot{\xi}_{\se} &=& \big(I \kron A - \hat{L}_{\se}\kron BC\big)\xi_{\se} + (\hat{G}_{\se}\kron B)u \\
\hat{y}_{\se} &=& (\hat{H}_{\se}\kron C)\xi_{\se},
\end{array}\right.\label{eqn_sysrede}
\end{align}
with $\hat{L}_{\se} \defl \hat{E}^{\T}\hat{F}$ the reduced-order edge Laplacian for the graph $\hat{\Gcal}$. In (\ref{eqn_sysrede}), the matrices $\hat{G}_{\se}$ and $\hat{H}_{\se}$ are given as $\hat{G}_{\se} = \hat{E}^{\T}W^{\T}G$ and $\hat{H}_{\se} = HV\hat{F}(\hat{E}^{\T}\hat{F})^{-1}$, respectively. Moreover, after expressing (\ref{eqn_sysrede}) in new coordinates $\xi_{\sf} = ((\hat{E}^{\T}\hat{F})^{-1}\kron I)\xi_{\se}$, the reduced-order dual edge system $\sysred_{\sf}$ is obtained. Similar to the high-order counterpart in (\ref{eqn_sysf}), it has the same form as the reduced-order edge system (\ref{eqn_sysrede}) with new external input matrix $\hat{G}_{\sf} = (\hat{E}^{\T}\hat{F})^{-1}\hat{E}^{\T}W^{\T}G$ and external output matrix $\hat{H}_{\sf} = HV\hat{F}$.

The edge system (\ref{eqn_sys}) and dual edge system (\ref{eqn_sysf}) play a crucial role in the identification of the most suitable vertices for clustering. After introducing the partitioning
\begin{align}
\!\!\!L_{\se} = \left[\begin{array}{cc} L_{\se,11} & L_{\se,12} \\ L_{\se,21} & L_{\se,22} \end{array}\right]\!,
G_{\se} = \left[\begin{array}{c} G_{\se,1} \\ G_{\se,2} \end{array}\right]\!,
H_{\sf} = \left[\begin{array}{cc} H_{\sf,1} & H_{\sf,2} \end{array}\right]\!,\!\!\!
\label{eqn_partitioning_Le_Ge_Hf}
\end{align}
it can be shown that their reduced-order counterparts are related to the high-order versions through (a partial) singular perturbation procedure, which will be shown to have desirable consequences.
\begin{lemma}\label{lem_edgedynamics_schur}
Consider the edge system $\sys_{\se}$ as in (\ref{eqn_syse}) and the dual edge system $\sys_{\sf}$ as in (\ref{eqn_sysf}) with the partitioned matrices (\ref{eqn_partitioning_Le_Ge_Hf}) and the reduced-order counterparts $\sysred_{\se}$ and $\sysred_{\sf}$ obtained after application of the projection (\ref{eqn_projection_VW}). Then,
\begin{align}
\hat{L}_{\se} &= L_{\se,11} - L_{\se,12}L_{\se,22}^{-1}L_{\se,21}, \label{eqn_lem_edgedynamics_Leschur} \\
\hat{G}_{\se} &= G_{\se,1} - L_{\se,12}L_{\se,22}^{-1}G_{\se,2}, \label{eqn_lem_edgedynamics_Geschur}\\
\hat{H}_{\sf} &= H_{\sf,1} - H_{\sf,2}L_{\se,22}^{-1}L_{\se,21}. \label{eqn_lem_edgedynamics_Hfschur}
\end{align}
\end{lemma}
\vspace{2mm}
\begin{proof}
First, the relation (\ref{eqn_lem_edgedynamics_Leschur}) will be proven. Thereto, it is noted that $\hat{L}_{\se}$ can be written as
\begin{align}
\!\!\!\hat{L}_{\se} = E_{00}^{\T}F_{00} +\! \tfrac{1}{w_{ij}+w_{ji}}(E_{i0} \!+\! E_{j0})^{\T}(w_{ji}F_{i0} \!+\! w_{ij}F_{j0}),\!\! \label{eqn_lem_edgedynamics_schur_proof_step1}
\end{align}
which follows from the definition $\hat{L}_{\se} = \hat{E}^{\T}\hat{F}$ and the definitions of $\hat{E}$ and $\hat{F}$ in (\ref{eqn_lem_Lhat_factor_Ehat_Fhat}). Next, from (\ref{eqn_Ledef}) and the partitioned matrices $E$ and $F$ in (\ref{eqn_partitioning_EF}) it can be concluded that
\begin{align}
L_{\se,11} = E_{00}^{\T}F_{00} + E_{i0}^{\T}F_{i0} + E_{j0}^{\T}F_{j0}
\label{eqn_lem_edgedynamics_schur_proof_step2}
\end{align}
and that the product $L_{\se,12}L_{\se,22}^{-1}L_{\se,21}$ reads
\begin{align}
L_{\se,12}L_{\se,22}^{-1}L_{\se,21} &= (E_{i0}^{\T}F_{il} + E_{j0}^{\T}F_{jl})\big(E_{il}^{\T}F_{il} + E_{jl}^{\T}F_{jl}\big)^{\!-1}\nonumber\\
&\phantom{=}\quad\times(E_{il}^{\T}F_{i0} + E_{jl}^{\T}F_{j0}).
\label{eqn_lem_edgedynamics_schur_proof_step3}
\end{align}
At this point, it is recalled that the $l$-th column of $E$ (and $F$) represents the edge connecting vertices $i$ and $j$, such that $E_{il}\in\{-1,1\}$, $E_{jl} = -E_{il}$ and $F_{il} = w_{ij}E_{il}$, $F_{jl} = w_{ji}E_{jl}$. The substitution of these relations in (\ref{eqn_lem_edgedynamics_schur_proof_step3}) leads to
\begin{align}
L_{\se,12}L_{\se,22}^{-1}L_{\se,21} = \tfrac{1}{w_{ij} + w_{ji}} (w_{ij}E_{i0} - w_{ji}E_{j0})^{\T}(F_{i0} - F_{j0}), \nonumber
\end{align}
such that
\begin{align}
&\hspace{-1mm}L_{\se,11} - L_{\se,12}L_{\se,22}^{-1}L_{\se,21}  = E_{00}^{\T}F_{00} \nonumber\\
&\hspace{13mm} + \left(\tfrac{w_{ji}}{w_{ij}+w_{ji}}\right)E_{i0}^{\T}F_{i0}
+  \left(\tfrac{w_{ij}}{w_{ij}+w_{ji}}\right)E_{j0}^{\T}F_{j0} \nonumber\\
&\hspace{13mm} + \left(\tfrac{w_{ij}}{w_{ij}+w_{ji}}\right)E_{i0}^{\T}F_{j0}
+  \left(\tfrac{w_{ji}}{w_{ij}+w_{ji}}\right)E_{i0}^{\T}F_{i0}.\!\!
\label{eqn_lem_edgedynamics_schur_proof_step5}
\end{align}
It can be checked that (\ref{eqn_lem_edgedynamics_schur_proof_step5}) equals (\ref{eqn_lem_edgedynamics_schur_proof_step1}), which proves the relation (\ref{eqn_lem_edgedynamics_Leschur}) in the statement of the lemma.

The relations (\ref{eqn_lem_edgedynamics_Geschur}) and (\ref{eqn_lem_edgedynamics_Hfschur}) can be proven similarly.
\end{proof}

\subsection{Synchronization preservation and multi-step reduction}
The one-step reduced-order system $\sysred$ as in (\ref{eqn_sysred}) preserves the property of synchronization, as formalized as follows.
\begin{theorem}\label{thm_sysred_sync}
Consider the networked system $\sys$ as in~(\ref{eqn_sys}) satisfying Assumption~\ref{ass_coupling} and let $\sysred$ as in (\ref{eqn_sysred}) be an approximation obtained by projection. Then, any trajectory of\/ $\sysred$ for $u=0$ satisfies (for all $i,j\in\hat{\Vcal}\defl\{1,2,\ldots,\bar{n}-1\}$)
\begin{align}
\lim_{t\rightarrow\infty} \big(\xi_i(t) - \xi_j(t)\big) = 0.
\end{align}
\end{theorem}
\vspace{2mm}
\begin{proof}
By Lemma~\ref{lem_Lhat_factor}, the reduced-order graph $\hat{\Gcal}$ characterized by $\hat{L}$ satisfies all statements of Assumption~\ref{ass_coupling}.
As a result, synchronization follows directly from Theorem~\ref{thm_sys_sync}.
\end{proof}
Up to this point, a one-step reduction has been considered. However, the results in Lemma~\ref{lem_edgedynamics_schur} can be shown to have the following important consequence.
\begin{theorem}\label{thm_sysred_genedgegrams}
Consider the networked system $\sys$ as in (\ref{eqn_sys}) satisfying Assumption~\ref{ass_coupling} and the reduced-order networked system $\sysred$ as in (\ref{eqn_sysred}). Assume that the generalized edge controllability Gramian $\tilde{\Pi}^{\sc}$ and generalized edge observability Gramian $\tilde{\Pi}^{\so}$ exist and consider (\ref{eqn_genedgecongram_diagonal})-(\ref{eqn_genedgeobsgram_diagonal}). Then,
\begin{enumerate}
  \item $\tilde{\Pi}_1^{\sc}\defl\diag\{\pi_1^{\sc},\ldots,\pi_{\bar{n}_{\se}-1}^{\sc}\}$ is a generalized edge controllability Gramian for the reduced-order system~$\sysred$;
  \item $\tilde{\Pi}_1^{\so}\defl\diag\{\pi_1^{\so},\ldots,\pi_{\bar{n}_{\se}-1}^{\so}\}$ is a generalized edge observability Gramian for the reduced-order system~$\sysred$.
\end{enumerate}
\end{theorem}
\begin{proof}
The theorem can be proven by following~\cite{fernando_1982}. In particular, after defining the projection matrix $T_c = [\begin{array}{cc} I & -L_{\se,12}L_{\se,22}^{-1} \end{array}]$, it can be applied to (\ref{eqn_genedgecongram_ineq}) to obtain
\begin{align}
&T_c(L_{\se}\tilde{\Pi}^{\sc} + \tilde{\Pi}^{\sc}L_{\se}^{\T} - G_{\se}G_{\se}^{\T})T_c^{\T} \nonumber\\
&\qquad= \hat{L}_{\se}\tilde{\Pi}_1^{\sc} + \tilde{\Pi}_1^{\sc}\hat{L}_{\se}^{\T} - \hat{G}_{\se}\hat{G}_{\se}^{\T} \succcurlyeq 0,
\label{eqn_thm_sysred_genedgegrams_proof_step1}
\end{align}
where (\ref{eqn_partitioning_Le_Ge_Hf}) is used as well as (\ref{eqn_lem_edgedynamics_Leschur}) and (\ref{eqn_lem_edgedynamics_Geschur}). It can be seen that the right-hand side of the equality in (\ref{eqn_thm_sysred_genedgegrams_proof_step1}) characterizes a generalized edge controllability Gramian for the reduced-order system $\sysred$, proving the first part of the theorem. The observability counterpart can be proven similarly.
\end{proof}
The results in Theorem~\ref{thm_sysred_genedgegrams} show that the one-step reduced order system $\sysred$ as in (\ref{eqn_sysred}) can be characterized by the relevant parts of the original generalized edge Gramians. Combined with the observation that $\sysred$ satisfies Assumption~\ref{ass_coupling} (through Lemma~\ref{lem_Lhat_factor}), this implies that the one-step reductions can be repeatedly applied to obtain a clustered system of arbitrary order. Here, the preservation of synchronization as in Theorem~\ref{thm_sysred_sync} remains guaranteed.

\section{Illustrative example}\label{sec_example}
To illustrate the reduction procedure, a simplified thermal model of a corridor of six rooms is considered. Motivated by \cite{ma_2011}, each room is modeled as a two thermal-mass system, leading to the dynamics
\begin{align}
\begin{array}{rcl}
C_1\dot{T}_1^i &=& R^{-1}_{\text{\rm int}}(T_2^i-T_1^i) - R_{\text{\rm out}}^{-1}T_1^i + P_i, \\
C_2\dot{T}_2^i &=& R^{-1}_{\text{\rm int}}(T_1^i-T_2^i).
\end{array}
\label{eqn_example_roomdynamics}
\end{align}
Here, $T_1^i$ and $T_2^i$ represent the (deviations from the environmental) temperature of the fast thermal mass $C_1$ (representing the air in the room) and slow thermal mass $C_2$ (representing solid elements such as walls, floor and furniture), respectively (i.e., $C_2>C_1$). In (\ref{eqn_example_roomdynamics}), $R_{\text{\rm int}}$ is the thermal resistance between the slow and fast thermal masses in the room, whereas $R_{\text{\rm out}}$ represents the thermal resistance of the outer walls, hereby assuming that the environmental temperature is constant.
After choosing $x_i^{\T} = [\begin{array}{cc} T_1^i & T_2^i \end{array}]$, $v_i = P_i$ and $z_i = T_1^i$, it is readily checked that (\ref{eqn_example_roomdynamics}) can be written in the form (\ref{eqn_sysi}) with $Q = \diag\{C_1,C_2\}$, $J=0$,
\begin{align}
R = \frac{1}{R_{\text{\rm int}}C_1C_2}\left[\begin{array}{cc} C_2/C_1 & 1 \\ 1 & C_1/C_2\end{array}\right]
+ \frac{1}{R_{\text{\rm out}}C_1^2}\left[\begin{array}{cc} 1 & 0 \\ 0 & 0 \end{array}\right]
\end{align}
and $B = [\begin{array}{cc}C_1^{-1} & 0 \end{array}]^{\T}$. In (\ref{eqn_example_roomdynamics}), $v_i = P_i$ represents the power associated with external influences other than that of the outside temperature, being external inputs such as heaters and the heat exchange with neighbouring rooms.
In particular, a corridor of six rooms is considered, such that the coupling between the rooms is given by a path graph $\Gcal$ as in Figure~\ref{fig_example_clustering}. The interconnection can thus be written in the form (\ref{ass_coupling}), where the nonzero weights are given by the thermal resistances of the walls as $w_{ij} = w_{ji} = R_{\text{\rm wall}}^{-1}$. Moreover, the control of the third room is of interest. Assuming that the temperature of this room can be influenced (e.g., through heaters) and measured, it follows that $G = H^{\T} = e_3$. The parameters are taken as $C_1=4.35\cdot10^4$~J/K, $C_2=9.24\cdot10^6$~J/K, $R_{\text{\rm int}}=2.0\cdot10^{-3}$~K/W, $R_{\text{\rm out}}=23\cdot10^{-3}~$K/W and $R_{\text{\rm wall}}=16\cdot10^{-3}~$K/W.

At this point, it is remarked that the assumptions on the networked system (\ref{eqn_sys}) require that the (internal) dynamics of each room is equal. However, the thermal resistances of the walls separating the rooms are part of the interconnection (\ref{eqn_couplingi}) and can thus vary between rooms.

\begin{figure}
  \begin{center}
  \begin{tikzpicture}
    [auto,vertex/.style={circle,draw,inner sep=0mm,minimum size=7mm,fill=white}]
    \node (c1) at (-1.65,0) [rectangle,draw,rounded corners=4mm,minimum height=8mm,minimum width=19mm,fill=black!20] {};
    \node (c2) at (2.2,0) [rectangle,draw,rounded corners=4mm,minimum height=8mm,minimum width=30mm,fill=black!20] {};
    \node (v1) at (-2.2,0) [vertex] {$\sys_1$};
    \node (v2) at (-1.1,0) [vertex] {$\sys_2$};
    \node (v3) at (0,0) [vertex] {$\sys_3$};
    \node (v4) at (1.1,0) [vertex] {$\sys_4$};
    \node (v5) at (2.2,0) [vertex] {$\sys_5$};
    \node (v6) at (3.3,0) [vertex] {$\sys_6$};
    \draw (v1.east) -- (v2.west) node [above,align=center,midway]{$1$};
    \draw (v2.east) -- (v3.west) node [above,align=center,midway]{$2$};;
    \draw (v3.east) -- (v4.west) node [above,align=center,midway]{$3$};;
    \draw (v4.east) -- (v5.west) node [above,align=center,midway]{$4$};;
    \draw (v5.east) -- (v6.west) node [above,align=center,midway]{$5$};;
    \node (u) at (-0.8,-0.8) {$u$};
    \node (y) at (0.8,-0.8) {$y$};
    \draw [->] (u.north east) to [bend right=0] (v3.south west);
    \draw [<-] (y.north west) to [bend left=0] (v3.south east);
  \end{tikzpicture}
  \vspace{-3mm}
  \caption{Path graph representing a corridor and clusters after reduction.}
  \label{fig_example_clustering}
  \end{center}
\end{figure}
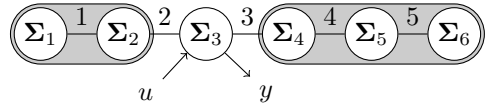

The generalized edge Gramians $\tilde{\Pi}^{\sc}$ and $\tilde{\Pi}^{\so}$ are computed by solving (\ref{eqn_genedgecongram_ineq})-(\ref{eqn_genedgeobsgram_ineq}), hereby minimizing their trace. Then, the computation of the products $\pi_i^{\sc}\pi_i^{\so}$ shows that edge $5$ has the smallest influence on the input-output behavior of the networked system, followed by edges $1$ and $4$.
Consequently, a three-step reduction leads to the clusters as in Figure~\ref{fig_example_clustering}, where it is noted that the rightmost cluster is formed in two steps. Thus, the two leftmost rooms as well as the three rightmost rooms are approximated as a single room each. However, the thermal resistances between these new approximated rooms and room three have been updated according to the projection (\ref{eqn_projection_VW}) (in three steps) to give a good representation of the original high-order model. Consequently, the wall thermal resistances are no longer equal throughout the (reduced-order) interconnection topology.
Finally, Figure~\ref{fig_example_frf} shows a comparison of the transfer functions of the original networked system $\sys$ as the reduced-order networked system $\sysred$, indicating a good approximation.

\section{Conclusions}\label{sec_conclusions}
A clustering-based approach towards model reduction of networks of interconnected passive subsystems is presented in this paper, hereby exploiting controllability and observability properties of the associated edge systems. The intuitive approach is shown to guarantee the preservation of synchronization properties.

\begin{figure}
  \begin{center}
    \xlabel{$f$ [1/hr]}
    \ylabel{$|T|$ [K/W]}
    \legend{T}{$T$}
    \legend{Tr}{$\hat{T}$}
    \includegraphics[scale=0.4]{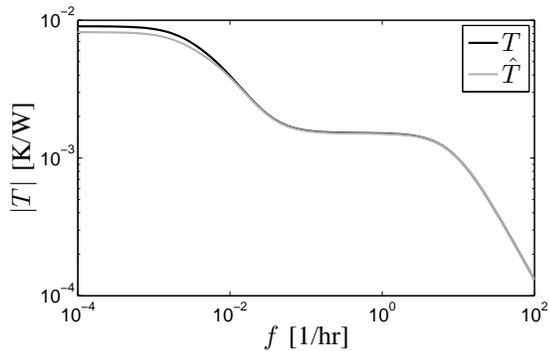}
    \vspace{-2mm}
    \caption{Comparison of the magnitude of the frequency response functions $T$ of $\sys$ and $\hat{T}$ of $\sysred$ for the configuration in Figure~\ref{fig_example_clustering}.}
    \label{fig_example_frf}
  \end{center}
\end{figure}

\bibliographystyle{plain}
\bibliography{report_v2}

\begin{thebibliography}{10}

\bibitem{book_antoulas_2005}
A.C. Antoulas.
\newblock {\em Approximation of large-scale dynamical systems}.
\newblock SIAM, Philadelphia, USA, 2005.

\bibitem{besselink_2013d}
B.~Besselink, H.~Sandberg, K.H. Johansson, and J.-I. Imura.
\newblock Controllability of a class of networked passive linear systems.
\newblock In {\em Proceedings of the 52nd IEEE Conference on Decision and
  Control, Florence, Italy}, pages 4901--4906, 2013.

\bibitem{besselink_2013b}
B.~Besselink, U.~Tabak, A.~Lutowska, N.~van~de Wouw, H.~Nijmeijer, D.J. Rixen,
  M.E. Hochstenbach, and W.H.A. Schilders.
\newblock A comparison of model reduction techniques from structural dynamics,
  numerical mathematics and systems and control.
\newblock {\em Journal of Sound and Vibration}, 332(19):4403--4422, 2013.

\bibitem{brewer_1978}
J.W. Brewer.
\newblock Kronecker products and matrix calculus in system theory.
\newblock {\em IEEE Transactions on Circuits and Systems}, CAS-25(9):772--781,
  1978.

\bibitem{book_dullerud_2000}
G.E. Dullerud and F.G. Paganini.
\newblock {\em A course in robust control theory - {A} convex approach},
  volume~36 of {\em Texts in Applied Mathematics}.
\newblock Springer-Verlag, New York, USA, 2000.

\bibitem{fernando_1982}
K.~Fernando and H.~Nicholson.
\newblock Singular perturbational model reduction of balanced systems.
\newblock {\em IEEE Transactions on Automatic Control}, AC-27(2):466--468,
  1982.

\bibitem{glover_1984}
K.~Glover.
\newblock All optimal {H}ankel-norm approximations of linear multivariable
  systems and their {$L^{\infty}$}-error bounds.
\newblock {\em International Journal of Control}, 39(6):1115--1193, 1984.

\bibitem{book_godsil_2001}
C.~Godsil and C.~Royle.
\newblock {\em Algebraic graph theory}, volume 207 of {\em Graduate Text in
  Mathematics}.
\newblock Springer-Verlag, New York, USA, 2001.

\bibitem{book_horn_1990}
R.A. Horn and C.R. Johnson.
\newblock {\em Matrix analysis}.
\newblock Cambridge University Press, Cambridge, United Kingdom, 1990.

\bibitem{ishizaki_2014}
T.~Ishizaki, K.~Kashima, J.-I. Imura, and K.~Aihara.
\newblock Model reduction and clusterization of large-scale bidirectional
  networks.
\newblock {\em IEEE Transactions on Automatic Control}, 59(1):48--63, 2014.

\bibitem{koshita_2006}
S.~Koshita, M.~Abe, and M.~Kawamata.
\newblock Gramian-preserving frequency transformation for linear
  continuous-time state-space systems.
\newblock In {\em Proceedings of the IEEE International Symposium on Circuits
  and Systems, Kos, Greece}, pages 453 -- 456, 2006.

\bibitem{ma_2011}
Y.~Ma, G.~Anderson, and F.~Borrelli.
\newblock A distributed predictive control approach to building temperature
  regulation.
\newblock In {\em Proceedings of the American Control Conference, San
  Francisco, USA}, pages 2089--2094, 2011.

\bibitem{book_mesbahi_2010}
M.~Mesbahi and M.~Egerstedt.
\newblock {\em Graph theoretic methods in multi-agent networks}.
\newblock Princeton University Press, New Jersey, USA, 2010.

\bibitem{monshizadeh_2013}
N.~Monshizadeh, H.L. Trentelman, and M.K. Camlibel.
\newblock Stability and synchronization preserving model reduction of
  multi-agent systems.
\newblock {\em Systems~\& Control Letters}, 62(1):1--10, 2013.

\bibitem{moore_1981}
B.C. Moore.
\newblock Principal component analysis in linear systems - controllability,
  observability, and model reduction.
\newblock {\em IEEE Transactions on Automatic Control}, AC-26(1):17--32, 1981.

\bibitem{ren_2005b}
W.~Ren and R.W. Beard.
\newblock Consensus seeking in multiagent systems under dynamically changing
  interaction topologies.
\newblock {\em IEEE Transactions on Automatic Control}, 50(5):655--661, 2005.

\bibitem{sandberg_2009}
H.~Sandberg and R.M. Murray.
\newblock Model reduction of interconnected linear systems.
\newblock {\em Optimal Control Applications and Methods}, 30(3):225--245, 2009.

\bibitem{book_vanderschaft_2000}
A.J. {\sortbib{Schaft}van der Schaft}.
\newblock {\em {$L_2$}-gain and passivity techniques in nonlinear control}.
\newblock Communications and Control Engineering Series. Springer-Verlag,
  London, Great Britain, second edition, 2000.

\bibitem{strogatz_2001}
S.H. Strogatz.
\newblock Exploring complex networks.
\newblock {\em Nature}, 410(6825):268--276, 2001.

\bibitem{willems_1972b}
J.C. Willems.
\newblock Dissipative dynamical systems part~{II}: {L}inear systems with
  quadratic supply rates.
\newblock {\em Archive for Rational Mechanics and Analysis}, 45(5):352--393,
  1972.

\bibitem{zelazo_2011}
D.~Zelazo and M.~Mesbahi.
\newblock Edge agreement: graph-theoretic performance bounds and passivity
  analysis.
\newblock {\em IEEE Transactions on Automatic Control}, 56(3):544--555, 2011.

\end{thebibliography}

\end{document}